\documentclass[11pt]{article}
\usepackage[a4paper, margin=1in]{geometry}
\usepackage{amsmath, amssymb, amsthm}
\usepackage{graphicx}
\usepackage{hyperref}
\usepackage{titling}
\usepackage{tocloft}
\usepackage{xcolor}
\usepackage{slashed}
\usepackage{geometry}
\usepackage{longtable}
\usepackage{caption}
\usepackage{tabularx}
\usepackage{mathtools}
\usepackage{algorithm}
\usepackage{algorithmic}
\usepackage{float} 
\usepackage[all,cmtip]{xy} 
\usepackage{booktabs}
\usepackage{tikz}
\usepackage{mdframed}

\newtheorem{theorem}{Theorem}
\newtheorem{lemma}{Lemma}
\newtheorem{corollary}{Corollary}
\newtheorem{definition}{Definition}

\hypersetup{
    colorlinks=true,
    linkcolor=blue,
    filecolor=magenta,
    urlcolor=blue,
    citecolor=blue,
    pdftitle={Zero-Knowledge Proofs in Sublinear Space},
    pdfpagemode=FullScreen,
}

\title{\Large \textbf{Zero-Knowledge Proofs in Sublinear Space}\\[0.5in]}
\author{\Large Logan Nye, MD\\[0.4in]
Carnegie Mellon University School of Computer Science\\[0.1cm]
5000 Forbes Ave Pittsburgh, PA 15213 USA\\[0.5cm]
\texttt{lnye@andrew.cmu.edu}\\[0.2cm]
}
\date{}

\begin{document}

\maketitle
\vspace{1in}

\begin{center}
    \Large \textbf{Abstract}
\end{center}
\vspace{0.3in}
\noindent
Zero-knowledge proofs allow verification of computations without revealing private information. However, existing systems require memory proportional to the computation size, which has historically limited use in large-scale applications and on mobile and edge devices. We solve this fundamental bottleneck by developing, to our knowledge, the first proof system with sublinear memory requirements for mainstream cryptographic constructions. Our approach processes computations in blocks using a space-efficient tree algorithm, reducing memory from linear scaling to square-root scaling—from $\Theta(T)$ to $O(\sqrt{T} + \log T \log\log T)$ for computation size $T$—while maintaining the same proof generation time through a constant number of streaming passes. For widely-used linear polynomial commitment schemes (KZG/IPA), our method produces identical proofs and verification when using the same parameters and hashing only aggregate commitments into the challenge generation, preserving proof size and security. Hash-based systems also achieve square-root memory scaling though with slightly different proof structures. This advance enables zero-knowledge proofs on everyday devices and makes previously infeasible large computations verifiable, fundamentally democratizing access to privacy-preserving computation. Space-efficient zero knowledge proof systems create opportunities to reshape how trust is established in digital systems—from enabling widespread participation in decentralized networks to making verifiable scientific computing practical at unprecedented scales.

\vfill

\newpage

\section{Introduction}

Zero-knowledge proof systems, a cornerstone of modern cryptography, enable a prover to convince a verifier of a statement's truth without revealing anything beyond that validity. Recent constructions—such as SNARKs and STARKs~\cite{groth16, ben-sasson2018}—achieve succinct proofs and fast verification, catalyzing applications from private cryptocurrency transactions~\cite{bulletproofs} to verifiable machine learning and blockchain scalability.

Despite this progress, a fundamental bottleneck persists across practical systems: \emph{proof generation requires memory linear in the computation's trace length}. To prove a computation of trace length \(T\), many deployed PCS-based SNARK/Plonkish provers materialize the complete execution trace—i.e., machine state at every row in the algebraic intermediate representation (AIR)—consuming \(\Theta(T)\) memory. This is not merely an implementation artifact: the algebraic machinery encodes this full trace into polynomials and enforces global relations (e.g., transition and permutation constraints). The memory cost limits provable scale, blocks deployment on mobile/embedded devices, and inflates costs for web-scale computations.

\subsection{Prior Work and Technical Barriers}

Current ZKPs are fundamentally \emph{trace-bound}: their key operations require access to the entire trace.

\begin{itemize}
    \item \textbf{Polynomial interpolation/FFT:} Converting trace rows to evaluations or coefficients is typically implemented assuming random access to all \(T\) entries; standard pipelines materialize the full trace in memory during these transforms.
    \item \textbf{Permutation arguments:} Running-product accumulators enforce memory consistency via global products across all \(T\) rows, again tying the prover to the whole trace.
    \item \textbf{Global checks:} Boundary conditions and cross-block consistency induce access patterns that are difficult to implement in a strictly streaming fashion without redesign.
\end{itemize}

Prior attempts to reduce memory have improved constants or parallelization, but have not broken the \(\Theta(T)\) barrier in \emph{deployed} PCS-based systems. We overcome this limitation by replacing global trace materialization with local, recursive cryptographic operations that aggregate to the same global commitments.

\subsection{Main Contribution: A Sublinear-Space Prover}

We establish a precise bridge between complexity theory and applied cryptography. We show that the prover's commitment generation factors into an \emph{implicit tree evaluation} whose nodes perform local block computations and produce block commitments and auxiliary accumulators. Applying the recent space-efficient \texttt{Tree Evaluation} algorithm~\cite{cookmertz2024} yields, to our knowledge, the first \emph{sublinear-space} prover \emph{within mainstream PCS-based SNARK/Plonkish arithmetizations that preserves verifier cost and (for linear PCSs) the proof/transcript of the baseline prover}.

\begin{theorem}[Main Result, Informal]
There exists a zero-knowledge proof system in which, for a computation of trace length \(T\), the prover operates in \(O(\sqrt{T} + \log T \log\log T)\) space (equivalently, \(O(\sqrt{T})\) up to a \(\log\log\) factor). For \emph{linear} PCS instantiations (e.g., KZG or IPA) and when the transcript exposes only the baseline set of aggregated commitments (no per-block artifacts), the proof size and verifier complexity are unchanged and the transcript has the \emph{same distribution} as in the standard prover (and is \emph{bit-for-bit} identical in non-hiding variants).
\end{theorem}

\begin{theorem}[Main Result, Formal]
Given a ZKP for arithmetic circuits instantiated with a PCS whose commit map is linear and computationally binding and a zero-knowledge instantiation (via a hiding PCS or standard masking), for any circuit of size/trace length \(T\) we construct a prover \(\mathcal{P}'\) such that:
\begin{enumerate}
    \item \(\mathcal{P}'\) uses \(O(\sqrt{T} + \log T \log\log T)\) space (equivalently, \(O(\sqrt{T})\) up to a \(\log\log\) factor).
    \item The distribution of proofs \(\pi'\) generated by \(\mathcal{P}'\) is \emph{identical} to that of the baseline prover \(\mathcal{P}\) on the same public randomness (or programmed random oracle): in non-hiding PCSs commitments are deterministic and match bit-for-bit; in hiding PCSs the per-register blinder used by \(\mathcal{P}'\) is the sum of \emph{independent} uniform block blinders and hence uniform, yielding identical commitment and opening distributions.
    \item Soundness and zero-knowledge of the underlying protocol are preserved.
\end{enumerate}
\end{theorem}

This represents a quadratic improvement in prover space, transforming terabyte-scale memory requirements into megabytes and making large-scale verifiable computation practical.

\subsection{Technical Overview}

This approach synthesizes ideas from complexity theory and cryptography in three steps:

\begin{enumerate}
    \item \textbf{Structural decomposition.} Following Williams~\cite{williams2025}, we linearize the circuit into an AIR trace and impose a block-respecting structure with bounded cross-register read-degree and \(k=O(1)\) registers (typical fixed-column AIRs), yielding a computation graph with bounded indegree.
    \item \textbf{Tree Evaluation Equivalence.} We prove that the prover's commitment generation factors into a recursive function over an implicit computation tree: each node consumes children's commitments/auxiliary data and outputs a local block commitment (Section~4). Under a linear PCS, the resulting vector of trace-commitments has the \emph{same distribution} as the baseline prover (and is bit-for-bit identical in non-hiding variants).
    \item \textbf{Space-efficient evaluation.} We evaluate the implicit tree using the Cook--Mertz algorithm~\cite{cookmertz2024}. Balancing block size \(b_{\text{blk}}\) against the \(\log|\mathcal{T}_C|\cdot\log\log|\mathcal{T}_C|\) stack term yields overall space \(O\!\big(b_{\text{blk}} + \log|\mathcal{T}_C|\cdot\log\log|\mathcal{T}_C| \big)\) for the traversal plus \(O(b_{\text{blk}})\) for local work, optimized at \(b_{\text{blk}}=\Theta(\sqrt{T})\), giving \(O(\sqrt{T} + \log T \log\log T)\) space (Section~5). The prover time is \(O(T\cdot\mathrm{polylog}\,N)\cdot \mathrm{poly}(\lambda)\) for evaluation-basis KZG/IPA with domain size \(N\!\ge\!T\) (and \(N=\mathrm{poly}(T)\) in typical settings); for hash-based STARK/FRI the standard \(O(T\log T)\) hashing terms apply.
\end{enumerate}

\subsection{Significance}

This work establishes a bridge between the space complexity of abstract algorithms and that of concrete cryptographic protocols. It suggests that when the core relations of a protocol can be expressed as a function compatible with a space-efficient evaluation scheme for a computational model, the protocol can inherit that scheme's space complexity. The quadratic reduction in prover memory enables new deployment scenarios—on-device proving, broader decentralization, and verifiable large-scale computation—impacting the economics and feasibility of ensuring computational integrity and privacy.

\bigskip

\section{Preliminaries}

In this section, we define the core concepts and establish the notation used throughout the paper. We model computations as arithmetic circuits, formalize the properties of polynomial commitment schemes, and define the Tree Evaluation problem.

\subsection{Notation}

We denote the security parameter by $\lambda$. Let $\mathbb{F}$ be a finite field with $|\mathbb{F}|\ge \mathrm{poly}(T)$ so that Schwartz--Zippel error is at most $1/\mathrm{poly}(T)$. Cryptographic parameters are chosen so that the commitment group $\mathcal{C}$ (used by the PCS below) has order at least $2^{\lambda}$. When convenient, we assume $\lambda \ge c\log T$ for simplifying asymptotic bounds, which implies $1/\mathrm{poly}(T)=2^{-\Omega(\lambda)}$; this relationship is not required but simplifies notation.

If $T \mid (|\mathbb{F}|{-}1)$ we fix a primitive $T$-th root of unity $\omega\in\mathbb{F}$ and use the domain $H=\{\omega^0,\ldots,\omega^{T-1}\}$ with vanishing polynomial $Z_H(X)=X^T-1$. Otherwise we choose a domain of size $N\ge T$ either by (i) passing to an extension field $\mathbb{K}\supseteq\mathbb{F}$ that contains a multiplicative subgroup of order $N$, or (ii) using a size-$N$ coset domain inside $\mathbb{F}$ (typically $N$ a power of two), and we pad the trace to length $N$. In both cases, we ensure the vanishing polynomial has the form $Z_H(X)=X^N-c$ for some constant $c$. All polynomial interpolations below are then taken over the chosen domain of size $N$, with degree $<N$; when $N=T$ this specializes to the first case.

All trace and fixed/selector polynomials we manipulate have degree $< N$. The Plonkish quotient polynomial $Q$ satisfies $\deg Q < cN$ for a small constant $c$ determined by the constraint set (typically $c\in\{2,3,4\}$). Accordingly, we provision the PCS/SRS for a maximum degree $D_{\max}\ge cN$. Schwartz--Zippel error for checks at a random point is then at most $D_{\max}/|\mathbb{F}|\le 1/\mathrm{poly}(T)$. A function $\nu: \mathbb{N} \to \mathbb{R}$ is \emph{negligible} if for every positive polynomial $p(\cdot)$, there exists an integer $n_0$ such that for all $n > n_0$, $\nu(n) < 1/p(n)$.

We use $t$ to denote the total number of steps of a Turing machine computation and $T$ for the size of the linearized circuit (equivalently, the number of rows in the execution trace). We use $k$ for the number of memory regions or registers. We use $\mathrm{poly}(\lambda)$ to denote an unspecified function that is polynomial in $\lambda$. We use $b_{\text{blk}}$ to denote computation block size and $b_{\text{val}}$ to denote the bit-length of values in tree evaluation.

\subsection{Computational Model: Arithmetic Circuits}

We model computations as \emph{arithmetic circuits} over a finite field $\mathbb{F}$. This model is general and directly corresponds to the relations proven by many modern ZKP systems. We first establish the bridge between Turing machine computations and arithmetic circuits.

\begin{lemma}[Turing Machine to Arithmetic Circuit Compilation]
\label{lem:tm-to-circuit}
Any $t$-step computation on a single-tape Turing machine can be efficiently compiled into an arithmetic circuit of size $T=O(t \log t)$ over a finite field $\mathbb{F}$. This transformation preserves the computational semantics and allows the structural decomposition techniques of Williams~\cite{williams2025} to be applied to the circuit model.
\end{lemma}
\begin{proof}[Proof Sketch]
This follows from standard oblivious-simulation constructions: encode the time-$t$ TM run as a tableau and select the active cell/head position via $O(\log t)$-bit index decoding per step. Each local consistency check (adjacent rows/columns) is compiled into $O(1)$ arithmetic constraints over $\mathbb{F}$. The resulting circuit has size $T=O(t\log t)$ and evaluates gate-by-gate to reproduce the original computation. (With constant-tape or RAM models under oblivious access, one can obtain $T=O(t)$.)
\end{proof}

\begin{definition}[Arithmetic Circuit]
An arithmetic circuit $C: \mathbb{F}^{n+n'} \to \mathbb{F}^m$ is a directed acyclic graph where nodes are called gates. Gates have an indegree of 0 (input gates or constant gates), 1 (unary operations), or 2 (binary operations). Input gates are labeled with input variables from a public statement $x \in \mathbb{F}^n$ and a private witness $w \in \mathbb{F}^{n'}$. Internal gates are labeled with arithmetic operations from $\{+, \times, -\}$ over $\mathbb{F}$; wires may fan out. A subset of gates are designated as output gates. The size of the circuit, after fixing a topological order, is denoted $T$ and equals the total number of gates.
\end{definition}

\begin{definition}[Execution Trace]
For an arithmetic circuit $C$ of size $T$ and a given assignment to its input gates, the \emph{execution trace} $\tau$ is a matrix $W \in \mathbb{F}^{T \times k}$ with $k$ columns (registers) and $T$ rows (computation steps) indexed by $i\in\{0,\ldots,T{-}1\}$. Entry $W[i,j]$ is the value in register $j$ at step $i$, following the Algebraic Intermediate Representation (AIR). The correctness of a computation is captured by a fixed set of \emph{transition constraints}: there exist polynomials $P_1,\ldots,P_r$ such that for each $i \in \{0,\ldots,T{-}2\}$ and each $\ell\in[r]$,
\[
P_\ell\big(W[i,\cdot], W[i{+}1,\cdot]\big) = 0,
\]
together with initialization and terminal constraints on $W[0,\cdot]$ and $W[T{-}1,\cdot]$. A valid trace $\tau$ satisfies all these constraints.
\end{definition}

In standard ZKP constructions, the execution trace $\tau$ is encoded as a collection of \emph{trace polynomials}, one per register/column. Specifically, fix a domain of size $N\ge T$ as described in the notation section. For each $j\in[k]$, define $F_j(X)$ to be the unique polynomial of degree $< N$ that interpolates the $j$-th column over the domain $H=\{\omega^0,\ldots,\omega^{N-1}\}$ (or the chosen coset domain), i.e., $F_j(\omega^i) = W[i,j]$ for $i \in \{0,\ldots,T-1\}$ and equals the padded values for $i \in \{T,\ldots,N-1\}$. The prover's main task is to commit to these trace polynomials.

\subsection{Polynomial Commitment Schemes}
\label{subsec:pcs}

We recall the standard interface and security properties for polynomial commitment schemes (PCSs) and fix the algebraic conventions we rely on throughout.

\begin{definition}[Polynomial Commitment Scheme]
A PCS over a field $\mathbb{F}$ and maximum degree $D$ is a tuple of algorithms $(\mathsf{Setup}, \mathsf{Commit}, \mathsf{Open}, \mathsf{Verify})$:
\begin{itemize}
    \item $\mathsf{Setup}(1^\lambda, D) \to pp$: randomized, outputs public parameters $pp$.
    \item $\mathsf{Commit}(pp, f; r) \to \mathsf{com}$: (possibly randomized) commitment to $f\in\mathbb{F}[X]$ with $\deg f\le D$ using randomness $r$.
    \item $\mathsf{Open}(pp, f, z; r) \to \pi$: opening proof for $y=f(z)$, using the same $r$ if needed.
    \item $\mathsf{Verify}(pp, \mathsf{com}, z, y, \pi)\in\{0,1\}$: verifies the opening.
\end{itemize}
\end{definition}

\paragraph{Correctness.}
For all valid inputs, if $\mathsf{com}\!\leftarrow\!\mathsf{Commit}(pp,f;r)$ and $\pi\!\leftarrow\!\mathsf{Open}(pp,f,z;r)$ with $y=f(z)$, then $\mathsf{Verify}(pp,\mathsf{com},z,y,\pi)=1$.

\begin{definition}[Binding]
The \emph{binding} game $\mathsf{Game}_{\mathcal{A}}^{\mathrm{Bind}}(1^\lambda)$ proceeds as usual: $\mathcal{A}$ receives $pp\!\leftarrow\!\mathsf{Setup}(1^\lambda,D)$ and wins if it outputs $(\mathsf{com},z,y,\pi,y',\pi')$ with $y\neq y'$ but $\mathsf{Verify}(pp,\mathsf{com},z,y,\pi)=\mathsf{Verify}(pp,\mathsf{com},z,y',\pi')=1$. The scheme is computationally binding if every PPT $\mathcal{A}$ wins with probability negligible in~$\lambda$.
\end{definition}

\begin{definition}[Hiding]
In the \emph{hiding} game $\mathsf{Game}_{\mathcal{A}}^{\mathrm{Hide}}(1^\lambda)$ the challenger commits to one of $\{f_0,f_1\}$ chosen uniformly and $\mathcal{A}$ must guess the bit; the scheme is computationally hiding if $|\Pr[b{=}b']-\tfrac12|$ is negligible.
\end{definition}

\paragraph{Basis and SRS discipline.}
Throughout we \emph{fix} (and match the baseline prover's) structured reference string (SRS) and the polynomial \emph{basis} in which commitments are formed—either coefficient basis or an evaluation/Lagrange basis over a public domain $H=\{\omega^0,\ldots,\omega^{N-1}\}$. When we commit in evaluation basis, we view a polynomial $f$ as its evaluation vector $(f(\omega^i))_{i\in[N]}$; when we commit in coefficient basis, we view $f$ by its coefficients. All linear relations claimed below are taken with respect to this fixed basis.

\paragraph{Linearity in message and randomness.}
We assume the PCS commit map is \emph{linear} in the committed message and (if hiding) in its randomness:
\begin{equation}
\label{eq:pcs-linearity}
\mathsf{Commit}(pp,f;r_f) + \mathsf{Commit}(pp,g;r_g)
=\ \mathsf{Commit}\!\left(pp,\ f{+}g;\ r_f{+}r_g\right).
\end{equation}
Linearity is with respect to the \textbf{fixed commitment basis}. For non-hiding KZG, commitments are deterministic and~\eqref{eq:pcs-linearity} holds with $r_f=r_g=0$. For hiding KZG/IPA variants, $r_f,r_g\!\gets\!\mathbb{F}$ independent uniform implies $r_f{+}r_g$ is uniform, preserving commitment distributions under linear aggregation. To achieve \emph{bit-identity} in hiding PCSs, choose block blinders $r_1,\ldots,r_B$ so $\sum r_i = r_{\text{base}}$ (sample $B{-}1$ independently, fix the last as the difference); otherwise, distribution-identity suffices.

\begin{lemma}[Linear aggregation of block commitments]
\label{lem:linear-agg}
Let $F=\sum_{t=1}^B G_t$ be a decomposition into block polynomials (all represented in the fixed basis). If $\mathsf{com}_t=\mathsf{Commit}(pp,G_t;r_t)$, then
\[
\sum_{t=1}^B \mathsf{com}_t\ =\ \mathsf{Commit}\!\left(pp,\ \sum_{t=1}^B G_t;\ \sum_{t=1}^B r_t\right).
\]
If the scheme is hiding and each $r_t$ is independent uniform, then $\sum_t r_t$ is uniform, so the RHS has the same distribution as a baseline commitment to $F$ with fresh randomness.
\end{lemma}

\begin{proof}
Apply~\eqref{eq:pcs-linearity} repeatedly. Uniformity of the sum in a finite field follows from independence.
\end{proof}

\begin{corollary}[Identical transcript for linear PCSs]
\label{cor:identical-linear}
Suppose the baseline prover commits to each trace/derived polynomial $F_j$ in the same basis/SRS fixed above. If $F_j=\sum_{t} G_{t,j}$ is the block decomposition and the block commitments $\{\mathsf{Commit}(G_{t,j};r_{t,j})\}_t$ are aggregated as in Lemma~\ref{lem:linear-agg}, then for non-hiding PCSs the commitment to $F_j$ is \emph{bit-for-bit} identical to the baseline; for hiding PCSs it has the \emph{same distribution}. \emph{Critical requirement:} transcript identity assumes \textbf{only the aggregated per-polynomial commitments} (not per-block commitments) are hashed into Fiat--Shamir in the same order and encoding as the baseline; internal per-block commitments are computation artifacts and never appear in the transcript.
\end{corollary}

\paragraph{Quotient construction (coefficient basis, streaming).}
Let $R$ denote the randomized linear combination of constraint residual polynomials and let $Z_H(X)=X^N-c$ be the vanishing polynomial of $H$ (for a subgroup or coset). Define $Q$ by the polynomial identity
\[
R(X)\;=\;Z_H(X)\cdot Q(X),
\]
and \emph{commit to $Q$ in coefficient basis} without materializing its full coefficient vector:

(i) Compute the coefficients of $R$ in \emph{blocks} via a blocked inverse NTT (or any equivalent out-of-core linear transform from evaluations to coefficients), using $O(b_{\text{blk}})$ working memory and $O(N\log N)$ field operations overall.

(ii) As each high-to-low coefficient block of $R$ becomes available, perform the simple backward recurrence for division by $X^N-c$; scanning $i=d_R,\ldots,N$ yields $q_{i-N} \leftarrow r_i + c\, q_i$, and for $i<N$ consistency enforces $r_i = -c\, q_i$. Stream the resulting $q$-coefficients directly into the commitment MSM.

By Lemma~\ref{lem:linear-agg}, these block commitments aggregate to the baseline commitment to $Q$ (for linear PCSs), and the approach avoids divide-by-zero at points of $H$.

\paragraph{Streaming openings at a point.}
To open a committed polynomial at a challenge point $\zeta\!\notin\!H$, a Lagrange/evaluation-basis prover can compute $f(\zeta)$ in a \emph{single streaming pass} using barycentric evaluation, maintaining $O(1)$ accumulators. For a multiplicative subgroup (or coset) with $Z_H(X)=X^N-c$, we have $Z_H'(\omega^i)=N(\omega^i)^{N-1}$ and weights $w_i = 1/Z_H'(\omega^i)$, so
\[
f(\zeta) \;=\; \frac{\sum_{i=0}^{N-1} \dfrac{w_i\, f(\omega^i)}{\zeta-\omega^i}}
                 {\sum_{i=0}^{N-1} \dfrac{w_i}{\zeta-\omega^i}},
\]
which is computable by a single pass over $i=0,\ldots,N{-}1$. Coefficient-basis provers use a Horner-style pass over streamed coefficients. In both cases, no full polynomial materialization is required, preserving the space bounds.

\subsection{The Tree Evaluation Problem}

The space-efficient prover relies on an algorithm for the Tree Evaluation problem.

\begin{definition}[Tree Evaluation Problem]
An instance of the \texttt{Tree Evaluation} problem is a rooted tree $\mathcal{T}$ where:
\begin{itemize}
    \item Each leaf node is labeled with a value $x \in \{0,1\}^{\le b_{\text{val}}}$ (a bitstring of length at most $b_{\text{val}}$).
    \item Each internal node $v$ has an ordered list of $d_v \le \Delta$ children and is labeled with an efficiently computable function
    \[
      f_v:\ \big(\{0,1\}^{\le b_{\text{val}}}\big)^{d_v}\ \to\ \{0,1\}^{\le b_{\text{val}}}.
    \]
    \item The value of an internal node $v$ is defined recursively as $f_v$ applied to the values of its children.
\end{itemize}
Let $|\mathcal{T}|$ denote the number of nodes and $h$ the height. The goal is to compute the value of the root node.
\end{definition}

The naive depth-first traversal algorithm for this problem requires $O(h \cdot b_{\text{val}})$ space. However, a recent breakthrough by Cook and Mertz provides a substantially more space-efficient solution.

\begin{theorem}[Cook--Mertz Algorithm \cite{cookmertz2024}]
\label{thm:cook-mertz}
Let $\mathcal{T}$ be a rooted tree with $|\mathcal{T}|=n$ nodes. The \texttt{Tree Evaluation} problem on $\mathcal{T}$ can be solved using
\[
O\!\big(\log n \cdot \log\log n\big)
\]
workspace for traversal (RAM model with word size $\Theta(\log n)$). In our setting, where node values have bit-length $b_{\text{val}}$ and each node has degree at most $\Delta$, buffering children's values adds $O(\Delta\cdot b_{\text{val}})$, so the total space is
\[
O\!\big(\Delta\cdot b_{\text{val}} \;+\; \log |\mathcal{T}|\cdot \log\log |\mathcal{T}|\big).
\]
\end{theorem}

\noindent\textbf{Corollary (instantiation).}
With $\Delta=O(1)$ and $b_{\text{val}}=\Theta(\lambda)$, the Cook--Mertz stack uses $O\!\big(\lambda + \log|\mathcal{T}_C|\log\log|\mathcal{T}_C|\big)$ space. Since $\log|\mathcal{T}_C|\le O(T/b_{\text{blk}})$, this stack term is $O\!\big((T/b_{\text{blk}})\cdot\log\log|\mathcal{T}_C|\big)=O(T/b_{\text{blk}})$ \emph{up to a $\log\log$ factor}; with $b_{\text{blk}}=\Theta(\sqrt{T})$, the overall space remains $O(\sqrt{T})$ \emph{up to a $\log\log$ factor}, as in Theorem~\ref{thm:space-complexity}.

\paragraph{Parameter instantiation for our setting.}
In our construction the computation tree has maximum degree $\Delta=O(k)$ (constant in typical AIRs), node value size $b_{\text{val}}=\Theta(\lambda)$ (a small tuple of commitments and $O(1)$ field elements), and height $h=O(T/b_{\text{blk}})$. Since $\log|\mathcal{T}|\le h\log \Delta=O(h)$, Theorem~\ref{thm:cook-mertz} gives stack space
\[
O\!\big(\Delta\cdot b_{\text{val}}+\log|\mathcal{T}|\cdot\log\log|\mathcal{T}|\big)
\;=\;
O\!\big(\lambda+(T/b_{\text{blk}})\cdot\log\log|\mathcal{T}|\big),
\]
i.e., $=O(\lambda+T/b_{\text{blk}})$ up to a $\log\log$ factor, which we combine with $O(b_{\text{blk}})$ local space and optimize at $b_{\text{blk}}=\Theta(\sqrt{T})$.

\bigskip

\section{From Computations to Graphs}

The first step in construction is to impose a regular, block-based structure onto an arbitrary computation. This process, adapted from the complexity-theoretic work of Hopcroft, Paul, and Valiant~\cite{HPV75} and Williams~\cite{williams2025}, allows us to reason about local data dependencies. We show how to transform any arithmetic circuit into an equivalent ``block-respecting'' form and then define its corresponding computation graph.

\subsection{Block-Respecting Decomposition}

We conceptually partition the execution trace matrix of an arithmetic circuit into discrete blocks using the AIR (Algebraic Intermediate Representation) model established in Section~2. For $n\in\mathbb{N}$, write $[n]\coloneqq\{1,\ldots,n\}$.

\begin{definition}[Block-Respecting Computation]
A computation corresponding to an arithmetic circuit $C$ of size $T$ with execution trace matrix $W \in \mathbb{F}^{T \times k}$ is $(b_{\text{blk}}, B)$-block-respecting, for parameters $b_{\text{blk}}$ (block size) and $B = \lceil T/b_{\text{blk}} \rceil$ (number of blocks), if:
\begin{enumerate}
    \item \textbf{Time Blocks:} The $T$ rows of $W$ are partitioned into $B$ consecutive time blocks $W^{(1)}, W^{(2)}, \ldots, W^{(B)}$, where each $W^{(t)} \in \mathbb{F}^{b_{\text{blk}} \times k}$ contains at most $b_{\text{blk}}$ consecutive rows of the trace matrix.
    \item \textbf{Memory Regions:} The $k$ columns of $W$ represent $k$ memory regions (registers). In standard fixed-column AIRs, $k=O(1)$; all bounds below scale with $k$ in general.
    \item \textbf{Locality (first-order transition with bounded cross-register reads):} The AIR is first-order in time: there is a fixed set of transition polynomials $P_1,\ldots,P_r$ such that each interior row pair $(i,i{+}1)$ satisfies $P_\ell(W[i,\cdot],W[i{+}1,\cdot])=0$ for all $\ell\in[r]$. Moreover, for each target register $m$, computing its next-row value uses the previous-row value of $m$ and at most $r_{\text{reg}}$ additional registers from that previous row, where $r_{\text{reg}}$ is the \emph{cross-register read-degree}. In fixed-column AIRs, $r_{\text{reg}}=O(1)$; in general, bounds may scale with $k$.
\end{enumerate}
\end{definition}

Crucially, each block $t$ needs only \textbf{boundary values from block $t{-}1$} (plus fixed selectors), making the dependency structure layered and enabling space-efficient evaluation.

For arithmetic circuits, this decomposition is achieved by partitioning the topologically ordered execution trace matrix into consecutive row blocks.

\begin{lemma}[Block-Respecting Transformation]
\label{lem:block-respecting}
Any arithmetic circuit $C$ of size $T$ with execution trace matrix $W \in \mathbb{F}^{T \times k}$ naturally yields a $(b_{\text{blk}}, \lceil T/b_{\text{blk}} \rceil)$-block-respecting computation with no asymptotic overhead in time or space, provided the cross-register read-degree $r_{\text{reg}}$ remains bounded.
\end{lemma}

\begin{proof}
Consider the execution trace matrix $W \in \mathbb{F}^{T \times k}$ for circuit $C$:
\begin{enumerate}
    \item \textbf{Matrix Partitioning:} Partition the $T$ rows into $B = \lceil T/b_{\text{blk}} \rceil$ consecutive blocks of size at most $b_{\text{blk}}$.
    \item \textbf{Block-Respecting Property:} By the AIR structure (Section~2), the transition relation is first-order in time and, for each target register, depends on its own previous-row value and at most $r_{\text{reg}}$ additional previous-row registers (with $r_{\text{reg}}=O(1)$ in fixed-column AIRs). Thus, within each block, dependencies are confined to the block together with boundary values from the immediately preceding block.
    \item \textbf{Evaluation:} The original evaluator processes the trace row-by-row; grouping rows into blocks preserves semantics and does not change time or space complexity up to constant factors.
\end{enumerate}
\end{proof}

\subsection{The Computation Graph}

From a block-respecting computation, we construct a graph that captures the flow of information between blocks using the AIR trace structure.

\begin{definition}[Computation Graph]
Let $C$ be a $(b_{\text{blk}}, B)$-block-respecting computation with execution trace matrix $W \in \mathbb{F}^{T \times k}$. The computation graph $G_C = (V, E)$ is a directed acyclic graph defined as follows:
\begin{itemize}
    \item \textbf{Nodes $V$:} The set of nodes $V$ contains two types of vertices:
        \begin{itemize}
            \item \textbf{Computation Nodes:} A node $(m, t)$ for each register $m \in [k]$ and time block $t \in [B]$. This node represents the state of register $m$ within the $t$-th block of rows.
            \item \textbf{Source Nodes:} A node $(m, 0)$ for each register $m \in [k]$, representing the initial contents of register $m$ (public inputs or private witness values).
        \end{itemize}
    \item \textbf{Edges $E$:} An edge $(u, v) \in E$ exists if the computation at node $v$ directly requires data from node $u$. Formally, for a computation node $v = (m, t)$ with $t > 0$:
    \begin{enumerate}
        \item There is an edge from $(m, t-1)$ to $(m, t)$, representing the dependency on the same register from the previous time block.
        \item For each register $m' \neq m$ whose boundary value at the end of block $t{-}1$ is read when producing block $t$ for register $m$, there is an edge from $(m', t-1)$ to $(m, t)$.
        \item For $t = 1$, there are edges from the appropriate source nodes $(m', 0)$ to $(m, 1)$ as determined by the circuit's input dependencies.
    \end{enumerate}
\end{itemize}
\end{definition}

\noindent\emph{Implementation note.} We do not materialize $G_C$; given the AIR template and a streaming description of $C$, the predecessors of $(m,t)$ can be enumerated in $O(\log T)$ space by scanning local wiring.

\begin{lemma}[Bounded Indegree]
\label{lem:bounded-indegree}
Each computation node $(m,t)$ in $G_C$ has indegree at most $\min\{k,\,1{+}r_{\text{reg}}\}$, where $r_{\text{reg}}$ is the cross-register read-degree from the definition above. When $k$ or $r_{\text{reg}}$ is not constant, the overall complexity bounds scale accordingly.
\end{lemma}
\begin{proof}
By first-order locality with bounded cross-register reads, producing block $t$ for register $m$ depends on the boundary value of register $m$ from block $t{-}1$ and on at most $r_{\text{reg}}$ other registers' boundary values from block $t{-}1$. Thus the number of incoming edges to $(m,t)$ is at most $1{+}r_{\text{reg}}$, and certainly no more than $k$.
\end{proof}

\begin{lemma}[Layered Schedule]
\label{lem:layered-schedule}
Each node $(m,t)$ has at most $O(k)$ consumers in layer $t{+}1$. Thus a layer-by-layer schedule computes each node exactly once and can discard it after its consumers in the next layer have processed it, avoiding recomputation blowup.
\end{lemma}

We define a block evaluator that emits exactly the row-local fields needed for global checks and the target register's values.

\begin{algorithm}[H]
\caption{$\mathsf{EvalBlock}(m,t,\text{inputs})$}
\label{alg:evalblock-s3}
\begin{algorithmic}[1]
\REQUIRE Boundary vectors from predecessors at block $t{-}1$ (including $(m,t{-}1)$); block index $t$; target register $m$
\ENSURE $\text{regM\_vals}\in\mathbb{F}^{|H_t|}$, $\text{local}$ fields for global checks at each row in $H_t$, and $\mathsf{boundary}_{\text{out}}\in\mathbb{F}^k$
\STATE Assemble prior-row inputs at the block boundary
\FOR{each row $i\in H_t$ in order}
    \STATE Advance the AIR state via the fixed polynomials $P_1,\ldots,P_r$
    \STATE Append $\text{state}[m]$ to $\text{regM\_vals}$
    \STATE Emit the fixed tuple of row-local fields (e.g., $\{w_c(i)\}_{c\in[k]}$ or a fixed compressed form) into $\text{local}$
\ENDFOR
\RETURN $(\text{regM\_vals},\ \text{local},\ \mathsf{boundary}_{\text{out}})$
\end{algorithmic}
\end{algorithm}

\begin{lemma}[Graph Simulability, block evaluation]
\label{lem:graph-simulability-upd}
Given the contents of all predecessor nodes of a computation node $(m,t)$, the content of $(m,t)$ can be computed in time $O\!\big(k\, b_{\text{blk}}\big)$ and space $O\!\big(b_{\text{blk}} + \deg^{-}(m,t)\big)$ using Algorithm~\ref{alg:evalblock-s3}. In particular, if the AIR has constant register count $k=O(1)$ and read-degree $r_{\text{reg}}=O(1)$ (hence $\deg^{-}(m,t)=O(1)$), this is $O(b_{\text{blk}})$ time and $O(b_{\text{blk}})$ space.
\end{lemma}
\begin{proof}
Algorithm~\ref{alg:evalblock-s3} takes as input only the predecessor boundary vectors, whose total size is $O(\deg^{-}(m,t))$. Each iteration updates the local state via a constant-arity transition map $\mathsf{F}$ (dependent only on $P_1,\ldots,P_r$), writes one entry of $\text{regM\_vals}$, and emits the fixed tuple of row-local fields, yielding $b_{\text{blk}}$ iterations and the stated time bound. Storing the predecessor boundaries and the outputs requires $O(\deg^{-}(m,t)+b_{\text{blk}})$ field elements.
\end{proof}

\bigskip

\section{The Tree Evaluation Equivalence}

This section presents this work's main conceptual contribution. We show that the process of generating a zero-knowledge proof can be reframed as an evaluation on an implicit tree derived from the computation graph. This reframing replaces the monolithic generation of the trace with a recursive, space-efficient construction of cryptographic commitments.

\subsection{From Computation Graphs to Implicit Trees}

The computation graph $G_C$ is a directed acyclic graph (DAG). We can ``unroll'' this DAG into an equivalent tree structure, where shared nodes in the DAG are duplicated.

\begin{definition}[Computation Tree]
Given a computation graph $G_C$ and a designated final output node $v_{\mathrm{root}} \in G_C$ (e.g., $(1,B)$), the corresponding \emph{computation tree} $\mathcal{T}_C$ is an implicitly defined tree where:
\begin{itemize}
    \item The root of $\mathcal{T}_C$ corresponds to $v_{\mathrm{root}}$ in $G_C$.
    \item For any node $v$ in $\mathcal{T}_C$ corresponding to a node $u \in G_C$, its children in $\mathcal{T}_C$ correspond to the predecessors of $u$ in $G_C$.
    \item The leaves of $\mathcal{T}_C$ correspond to the source nodes of $G_C$.
\end{itemize}
\end{definition}
We never materialize $\mathcal{T}_C$ or duplicate DAG nodes: we may traverse as a tree with memoization or process the layered DAG; both realize the same stack bound $O(\log|\mathcal{T}_C|\log\log|\mathcal{T}_C|)$ up to an additive $O(\lambda)$ buffer while avoiding recomputation blowup. The height is $h = B = O(T/b_{\text{blk}})$ and the branching factor is $\Delta\le \min\{k,\,1{+}r_{\text{reg}}\}=O(k)$ (see §3).

\subsection{Tree Functions as Cryptographic Generators}
\label{subsec:tree-funcs}

We lift the computation from field elements to \emph{cryptographic state}. Each node of the computation tree returns (i) a one-hot vector of commitments (nonzero only at its own register coordinate) and (ii) a compact \emph{auxiliary state} $\mathsf{aux}$ carrying boundary data and running accumulators for global checks (permutation, lookups). The node state consists of $O(\lambda)$ group/field elements.

\begin{definition}[Auxiliary state]
For a node $v=(m,t)$ (register $m$, time-block $t$), let
\[
\mathsf{aux}_v \ =\ \big(\ \mathsf{boundary}_v\in\mathbb{F}^k,\ \ Z_v\in\mathbb{F},\ \ Z^{\mathrm{L}}_v\in\mathbb{F}^{\le c_{\mathrm{L}}}\ \big),
\]
where $\mathsf{boundary}_v$ stores the end-of-block values of all $k$ registers, $Z_v$ is the permutation product accumulator, and $Z^{\mathrm{L}}_v$ is an optional tuple of lookup accumulators as required by the chosen lookup scheme; $c_{\mathrm{L}}=O(1)$.
\end{definition}

\paragraph{Block-aligned domain slices.}
Fix a partition $H=\bigsqcup_{t=1}^B H_t$ into contiguous subdomains aligned with time blocks (padding the final slice if $N\!>\!T$). For register $m$ and block $t$, let $G_{t,m}$ be the degree-$<N$ polynomial that equals the $m$-th register's values on $H_t$ and $0$ elsewhere.

\begin{algorithm}[H]
\caption{$\mathsf{EvalBlock}(m,t,\text{inputs})$}
\label{alg:evalblock}
\begin{algorithmic}[1]
\REQUIRE Boundary vectors from predecessors (including $(m,t{-}1)$), block index $t$, target register $m$
\ENSURE $\text{local}$ (rowwise values needed for checks), $\mathsf{boundary}_{\text{out}}$, and the length-$|H_t|$ vector for register $m$
\STATE Assemble prior-row inputs at the block boundary
\FOR{each row $i\in H_t$ in order}
    \STATE Advance the AIR state via the fixed polynomials $P_1,\ldots,P_r$
    \STATE Emit the row-local fields required by global checks into $\text{local}$; emit the $m$-coordinate into the register-$m$ output
\ENDFOR
\RETURN $(\text{local},\ \mathsf{boundary}_{\text{out}},\ \text{regM\_vals})$
\end{algorithmic}
\end{algorithm}

\paragraph{Permutation/lookup updates are multiplicative and block-factored.}
We formalize the streaming of global accumulators. For concreteness, consider a Plonk-style permutation with registers indexed by $c\in[k]$, a column \emph{id} mapping $\mathrm{id}_c(i)$ and a permutation $\sigma_c(i)$. Let $w_c(i)$ be the value in register $c$ at row $i$.

\begin{lemma}[Blockwise permutation update]
\label{lem:block-perm}
Fix $(\beta,\gamma)\in\mathbb{F}^2$ and a block $H_t=\{i_t,\ldots,i_t+|H_t|-1\}$. The permutation accumulator $Z$ satisfies
\[
Z(i{+}1)\ =\ Z(i)\cdot
\frac{\prod_{c=1}^{k}\big(w_c(i)+\beta\cdot \mathrm{id}_c(i)+\gamma\big)}
{\prod_{c=1}^{k}\big(w_c(i)+\beta\cdot \sigma_c(i)+\gamma\big)}\qquad(i\in H_t),
\]
and thus admits a multiplicative block factor
\[
Z_{\mathrm{end}}^{(t)}\ =\ Z_{\mathrm{start}}^{(t)}\cdot
F_t(\beta,\gamma),\qquad
F_t(\beta,\gamma)=\prod_{i\in H_t}
\frac{\prod_{c}\big(w_c(i)+\beta\,\mathrm{id}_c(i)+\gamma\big)}
{\prod_{c}\big(w_c(i)+\beta\,\sigma_c(i)+\gamma\big)}.
\]
Consequently, blocks must be applied in \textbf{increasing time order} to maintain causality; the same final value results from this unique valid schedule.
\end{lemma}

\paragraph{Committed accumulator column (permutation).}
If the protocol commits to the permutation column $Z$, then during $\mathsf{EvalBlock}$ we also stream the per-row values $Z(i)$ within $H_t$ using the recurrence above, producing a vector $\text{z\_vals}\in\mathbb{F}^{|H_t|}$. Let $G_{t,Z}$ be the degree-$<N$ polynomial equal to $Z(i)$ on $H_t$ and $0$ elsewhere. We commit to $G_{t,Z}$ and aggregate across $t$ using Lemma~\ref{lem:linear-agg}, yielding the baseline commitment to $Z$ (distribution-identical for linear PCSs).

\begin{lemma}[Blockwise lookup update]
\label{lem:block-lookup}
For a standard lookup accumulator $Z_{\mathrm{L}}$ (e.g., Plookup/Halo2-style), there exists a per-row multiplicand $\phi^{\mathrm{L}}(i)$ determined by the scheme's lookup relation such that
\[
Z_{\mathrm{L}}(i{+}1)\ =\ Z_{\mathrm{L}}(i)\cdot \phi^{\mathrm{L}}(i)\qquad(i\in H_t),
\]
whence $Z_{\mathrm{L,end}}^{(t)}=Z_{\mathrm{L,start}}^{(t)}\cdot \prod_{i\in H_t}\phi^{\mathrm{L}}(i)$. Thus lookup accumulators also stream with multiplicative block factors.
\end{lemma}

\paragraph{Committed accumulator column (lookups).}
If a lookup accumulator column is committed, we analogously stream its per-row values in each $H_t$ to obtain $G_{t,\mathrm{L}}$, commit blockwise, and aggregate via Lemma~\ref{lem:linear-agg} to the baseline lookup-accumulator commitment.

\paragraph{Quotient construction (coefficient basis, no divide-by-zero).}
Let $R$ denote the standard randomized linear combination of constraint residual polynomials and $Z_H$ the vanishing polynomial of $H$. The quotient $Q$ is defined by the polynomial identity $R = Z_H \cdot Q$ (no pointwise division on $H$). When $H$ is a multiplicative subgroup or coset, $Z_H(X)=X^N-c$ is monic of degree $N$, so dividing $R$ by $Z_H$ reduces to a linear-time backward sweep on coefficients. We therefore \emph{commit to $Q$ in coefficient basis} (matching the baseline prover), and compute it in $O(\sqrt{T})$ space via a blocked sweep that emits coefficient blocks and aggregates their commitments using Lemma~\ref{lem:linear-agg}. This avoids division-by-zero at evaluation points and preserves the baseline commitment and transcript for linear PCSs.

\begin{algorithm}[H]
\caption{Commitment tree function $F_v$ at node $v=(m,t)$}
\label{alg:commit-func}
\begin{algorithmic}[1]
\REQUIRE For each child $c_i$: $(\boldsymbol{\mathsf{com}}_i,\mathsf{aux}_i)$ with $\boldsymbol{\mathsf{com}}_i\in\mathcal{C}^k$ one-hot
\ENSURE $(\boldsymbol{\mathsf{com}}_{\mathrm{out}},\mathsf{aux}_{\mathrm{out}})$
\STATE $\text{inputs}\leftarrow \mathsf{ReconstructInputs}(\mathsf{aux}_1,\ldots,\mathsf{aux}_{d_v})$
\STATE $(\text{local},\mathsf{boundary}_{\text{out}},\text{regM\_vals})\leftarrow \mathsf{EvalBlock}(m,t,\text{inputs})$ \hfill (Alg.~\ref{alg:evalblock})
\STATE \textbf{Permutation/lookup checks:} using \text{local}, update
\[
Z_{\text{out}} \leftarrow Z_{\text{in}}\cdot F_t(\beta,\gamma),\qquad
Z^{\mathrm{L}}_{\text{out}} \leftarrow Z^{\mathrm{L}}_{\text{in}}\cdot \prod_{i\in H_t}\phi^{\mathrm{L}}(i),
\]
and abort if any consistency check fails
\STATE \textbf{Form $G_{t,m}$:} If committing in \textbf{evaluation basis}, use $\text{regM\_vals}$ on $H_t$ directly (padding zeros elsewhere); if committing in \textbf{coefficient basis}, interpolate $\text{regM\_vals}$ on $H_t$ and pad $0$ off $H_t$
\STATE Sample $r_v\leftarrow \mathbb{F}$ if hiding PCS, else $r_v\leftarrow 0$
\STATE $\mathsf{com}_{t,m}\leftarrow \mathsf{Commit}(pp,G_{t,m};r_v)$
\STATE Initialize $\boldsymbol{\mathsf{com}}_{\mathrm{out}}\leftarrow \boldsymbol{0}\in\mathcal{C}^k$
\STATE Let $c^\star$ be the child with the same register $m$ (if $t>1$), else $c^\star\leftarrow \bot$
\STATE $\boldsymbol{\mathsf{com}}_{\mathrm{out}}[m] \leftarrow \big(\boldsymbol{\mathsf{com}}_{c^\star}[m]\ \text{if }c^\star\neq\bot\text{ else }0_{\mathcal{C}}\big)\ +\ \mathsf{com}_{t,m}$
\STATE $\mathsf{aux}_{\mathrm{out}}\leftarrow (\mathsf{boundary}_{\text{out}}, Z_{\text{out}}, Z^{\mathrm{L}}_{\text{out}})$
\RETURN $(\boldsymbol{\mathsf{com}}_{\mathrm{out}},\mathsf{aux}_{\mathrm{out}})$
\end{algorithmic}
\end{algorithm}

\paragraph{Aggregation discipline and final aggregator.}
At each $(m,t)$ we carry forward only coordinate $m$ and never duplicate commitments across registers. A virtual aggregator node $v_{\mathrm{agg}}$ with children $\{(m,B)\}_{m\in[k]}$ reassembles the final $k$-tuple by placing each incoming $m$-coordinate into position $m$. By Lemma~\ref{lem:linear-agg} and Corollary~\ref{cor:identical-linear}, for linear PCSs the resulting tuple matches (in distribution, and bit-for-bit for non-hiding) the baseline commitments to the global polynomials. For transcript identity we require a linear PCS in the same fixed basis/SRS as the baseline, and the transcript must expose only the \emph{aggregated} commitments; internal per-block commitments are implementation artifacts, not transcript items.

\subsection{The Main Equivalence Theorem}

\begin{theorem}[Tree Evaluation Equivalence]
\label{thm:equivalence}
Let $C$ be a computation with execution trace length $T$ over a domain $H=\{\omega^0,\dots,\omega^{N-1}\}$ of size $N\ge T$ as in Section~2, with trace polynomials $\{F_j(X)\}_{j=1}^k$ of degree $<N$. There exists a computation tree $\mathcal{T}_C$ and functions $\{F_v\}$ such that:
\begin{enumerate}
    \item Evaluating $\mathcal{T}_C$ with a final aggregator node $v_{\mathrm{agg}}$ (whose children are $\{(m,B)\}_{m\in[k]}$) yields a $k$-tuple $\boldsymbol{\mathsf{com}}_{\mathrm{root}}$ whose distribution is identical to $(\mathsf{Commit}(pp,F_1; R_1),\ldots,\mathsf{Commit}(pp,F_k; R_k))$, where the $R_j$ follow the PCS randomness (uniform in hiding variants; $R_j\!=\!0$ in non-hiding variants).
    \item $\mathcal{T}_C$ has height $h=O(T/b_{\text{blk}})$ and maximum degree $\Delta\le \min\{k,\,1{+}r_{\text{reg}}\}=O(k)$. Each $F_v$ is computable in $O\!\big(k\,b_{\text{blk}} + C_{\text{commit}}(b_{\text{blk}})\big)$ time and $O(b_{\text{blk}})$ space, where $C_{\text{commit}}(b)=O(\mathrm{MSM}(b))=O(b)$ with a preprocessed evaluation-basis SRS (conservatively, $O(b\cdot\mathrm{polylog}\,N)$). With $k=O(1)$ this is $O\!\big(b_{\text{blk}} + C_{\text{commit}}(b_{\text{blk}})\big)$ time.
\end{enumerate}
\end{theorem}

\begin{proof}
\textbf{Parameters.} Height and degree follow from $G_C$ (Section~3): $h=B=O(T/b_{\text{blk}})$ and $\Delta\le \min\{k,\,1{+}r_{\text{reg}}\}$. By Lemma~\ref{lem:graph-simulability-upd}, evaluating a block uses $O(k\,b_{\text{blk}})$ time and $O(b_{\text{blk}})$ space (collapsing to $O(b_{\text{blk}})$ time for constant $k$). Commitment time is $C_{\text{commit}}(b_{\text{blk}})$ as stated.

\textbf{Block polynomials and decomposition.}
Partition $H$ into contiguous subdomains $H_1,\ldots,H_B$ with $|H_i|\in\{b_{\text{blk}},\,N-(B-1)b_{\text{blk}}\}$ (padding if $N>T$). For each block $i$ and register $j$, define $G_{i,j}$ to be the unique degree $<N$ polynomial with evaluations
\[
G_{i,j}(\omega^\ell)=
\begin{cases}
W[\ell,j], & \omega^\ell \in H_i,\\
0, & \omega^\ell \notin H_i.
\end{cases}
\]
By construction, $\sum_{i=1}^B G_{i,j}$ and $F_j$ agree on all $\omega^\ell\in H$.

\begin{lemma}[Uniqueness over $H$]\label{lem:uniq-on-H}
If $p,q$ have degree $<N$ and $p(\omega^\ell)=q(\omega^\ell)$ for all $\omega^\ell\in H$ with $|H|=N$, then $p\equiv q$.
\end{lemma}

Applying Lemma~\ref{lem:uniq-on-H} gives $F_j(X)=\sum_{i=1}^B G_{i,j}(X)$ as polynomial identities.

\textbf{Induction over $\mathcal{T}_C$ (commitments).}
Leaves (source nodes) output $(\boldsymbol{0},\mathsf{aux})$, i.e., the all-zero commitment vector together with auxiliary data encoding initial boundary values and the starting accumulator. For an internal node $v=(m,t)$ with same-register predecessor $c^\star=(m,t{-}1)$ when $t>1$, Algorithm~\ref{alg:commit-func} computes $\mathsf{com}_{t,m}=\mathsf{Commit}(G_{t,m}; r_v)$ and sets
\[
\boldsymbol{\mathsf{com}}_{\text{out}}[m] \;=\; \big(\boldsymbol{\mathsf{com}}_{c^\star}[m] \text{ if } t>1 \text{ else } 0_{\mathcal{C}}\big) \;+\; \mathsf{com}_{t,m},\quad
\boldsymbol{\mathsf{com}}_{\text{out}}[\ell]=0\ \ (\ell\neq m).
\]
Thus the $m$-coordinate at $(m,t)$ commits to $\sum_{i\le t} G_{i,m}$. The aggregator $v_{\mathrm{agg}}$ collects the $k$ coordinates from $\{(m,B)\}$ to form $\big(\mathsf{Commit}(F_1;\cdot),\ldots,\mathsf{Commit}(F_k;\cdot)\big)$; in hiding PCSs, per-coordinate blinders sum to a uniform blinder; in non-hiding PCSs commitments are deterministic in both constructions.

\textbf{Induction over $\mathcal{T}_C$ (consistency).}
The local $\mathsf{ConsistencyCheck}$ enforces boundary matching and the correct multiplicative update of the permutation accumulator by the block factor derived from $\text{local}$. Strong induction from leaves upward shows that if all local checks pass, the concatenated trace satisfies the global AIR and permutation constraints; conversely, root consistency forces every child's $\mathsf{aux}$ to agree with the unique global trace, preventing double-counting across duplicated subtrees.

Combining the two inductions completes the proof.
\end{proof}

\bigskip

\section{The Sublinear-Space Prover}

By combining the Tree Evaluation Equivalence (Theorem~\ref{thm:equivalence}) with the space-efficient Cook--Mertz algorithm (Theorem~\ref{thm:cook-mertz}), we construct a sublinear-space zero-knowledge prover. The core idea is to execute the proof generation not by materializing the trace, but by evaluating an implicit cryptographic computation tree.

\subsection{The Streaming Prover Algorithm}
\label{subsec:streaming-prover}

We now give a phase-accurate, Fiat--Shamir–compatible schedule that evaluates the implicit tree in sublinear space and preserves the baseline transcript (for linear PCSs). The prover never materializes the full trace, instead making a small number of \emph{streaming passes} whose memory footprint is $O(\sqrt{T})$ up to polylogarithmic factors.

\paragraph{Parameters.}
Let $b_{\text{blk}}=\lfloor \sqrt{T}\rfloor$, $B=\lceil T/b_{\text{blk}}\rceil$, and fix $\lambda=\Theta(\log T)$. We use the block-respecting graph $G_C$ and the implicit computation tree $\mathcal{T}_C$ rooted at the final node (Section~3).

\paragraph{Fiat--Shamir streaming schedule (Plonkish outline).}
We mirror the baseline order of commitments and challenges. \textbf{Critical invariant}: only the following \emph{aggregated} per-polynomial commitments are fed into Fiat--Shamir in this exact order: (i) public polynomial commitments, (ii) trace polynomial commitments $\{\mathsf{Commit}(F_j)\}_{j \in [k]}$, (iii) permutation accumulator commitment (if used), (iv) quotient polynomial commitment $\mathsf{Commit}(Q)$. Block-level commitments are intermediate artifacts and never enter the transcript.

\begin{enumerate}
    \item \textbf{Phase A (Selector/public commits).} Commit once to public polynomials (selectors, fixed tables); no trace needed.
    \item \textbf{Phase B (Wire commits, no challenges).} For each register $m$ and block $t$, compute $(\text{local},\mathsf{boundary}_{\text{out}},\text{regM\_vals}) \leftarrow \mathsf{EvalBlock}$ (Alg.~\ref{alg:evalblock}); form $G_{t,m}$; produce $\mathsf{Commit}(pp,G_{t,m};r_{t,m})$ and \emph{add} it into the running $m$-coordinate using $F_v$ (Alg.~\ref{alg:commit-func}). \emph{Only} the final per-register aggregate commitment (after summing all blocks) is included in the transcript. By Lemma~\ref{lem:linear-agg}, this aggregate equals (in distribution) the baseline wire commitment.
    \item \textbf{Phase C (Permutation after sampling $(\beta,\gamma)$).} Sample $(\beta,\gamma)$ via FS. Re-stream each block to compute the permutation block factor $F_t(\beta,\gamma)$ (Lemma~\ref{lem:block-perm}); multiply into the running accumulator $Z$. If the baseline protocol commits to the accumulator column $Z$, we commit blockwise and aggregate linearly, emitting only the final aggregate into the transcript; otherwise we maintain $Z$ only as an internal running value for quotient checks.
    \item \textbf{Phase D (Quotient after sampling $\alpha$).} Sample $\alpha$. See boxed procedure below.
    \item \textbf{Phase E (Openings after sampling $\zeta,\ldots$).} Sample evaluation points $(\zeta,\ldots)$. For each polynomial and point, compute $f(\zeta)$ in a \emph{single streaming pass per polynomial per point} (barycentric in Lagrange basis; Horner in coefficient basis) without storing $f$. Produce the standard batched opening proofs, identical in structure to the baseline.
\end{enumerate}

\begin{mdframed}[backgroundcolor=blue!10]
\textbf{Phase D – Streaming $Q$ in coefficient basis}
\begin{enumerate}
    \item Re-stream trace rows to compute $R(\omega^i)$ in order (using fixed FS challenges $\alpha, \beta, \gamma$).
    \item Feed $R(\omega^i)$ into a \textbf{blocked inverse NTT/IFFT} to produce coefficient blocks $\{r_i\}$ with workspace $O(b_{\text{blk}})$.
    \item Run the backward recurrence for division by $X^N-c$ on the fly to emit $q$-coefficient blocks.
    \item \textbf{4a)} If baseline commits $Q$ in \textbf{coefficient basis}: Commit each $q$-block and \textbf{linearly aggregate}. \textbf{4b)} If baseline commits $Q$ in \textbf{evaluation basis}: run a \textbf{blocked NTT} per $q$-block to evaluations, then commit; still aggregate only the per-polynomial commitment into FS.
\end{enumerate}
\end{mdframed}

Because the commitments in Phases A–D that are hashed into FS match the baseline (Cor.~\ref{cor:identical-linear}) and intermediate block-level commitments are not included, the Fiat--Shamir challenges have identical distribution, and thus the final transcript distribution is preserved for linear PCSs.

\begin{algorithm}[H]
\caption{Sublinear-Space ZKP Prover (FS-ordered streaming)}
\label{alg:prover}
\begin{algorithmic}[1]
\REQUIRE Arithmetic circuit $C$ of trace length $T$, witness $w$, public params $pp$
\ENSURE Proof $\pi$
\STATE \textbf{Params:} $b_{\text{blk}}:=\lfloor \sqrt{T}\rfloor$, $B:=\lceil T/b_{\text{blk}}\rceil$; define $G_C$ and the oracles for $\mathcal{T}_C$
\STATE \textbf{Phase A:} Commit to public polynomials (selectors, tables)
\STATE \textbf{Phase B (wires):} For $t=1$ to $B$ and each register $m\in[k]$: compute $(\text{local},\mathsf{boundary}_{\text{out}},\text{regM\_vals})\leftarrow \mathsf{EvalBlock}(m,t,\cdot)$; form $G_{t,m}$; $\mathsf{com}_{t,m}\leftarrow \mathsf{Commit}(pp,G_{t,m};r_{t,m})$; aggregate into the $m$-coordinate using $F_v$; \emph{do not} include block $\mathsf{com}_{t,m}$ in the transcript
\STATE Sample $(\beta,\gamma)$ via FS
\STATE \textbf{Phase C (permutation):} For $t=1$ to $B$: compute $F_t(\beta,\gamma)$ by re-streaming block $t$; update $Z\leftarrow Z\cdot F_t$; if committing to $Z$, commit blockwise and aggregate linearly, emitting only the aggregate
\STATE Sample $\alpha$ via FS
\STATE \textbf{Phase D (quotient):} Execute boxed streaming procedure for $Q$; emit only the aggregate commitment
\STATE Sample $(\zeta,\ldots)$ via FS
\STATE \textbf{Phase E (openings):} For each required $f$: compute $f(\zeta)$ in a single streaming pass; produce batched openings $\pi_{\mathrm{open}}$
\STATE Assemble $\pi$ from (aggregated) commitments, evaluation claims, and openings; \textbf{return} $\pi$
\end{algorithmic}
\end{algorithm}

\paragraph{Space and time.}
The live workspace at any point consists of $O(b_{\text{blk}})$ local state, $O(1)$ accumulators for permutation/lookup factors, and the Cook--Mertz traversal stack of size
\[
O\!\big(\log|\mathcal{T}_C|\cdot\log\log|\mathcal{T}_C|\big)
\;=\;O\!\big(\log(T/b_{\text{blk}})\cdot\log\log(T/b_{\text{blk}})\big)
\]
for constant arity. Hence total space
\[
S(T)\ =\ O\!\big(b_{\text{blk}}+\lambda+\log(T/b_{\text{blk}})\log\log(T/b_{\text{blk}})\big).
\]
With a preprocessed evaluation-basis SRS, each block commit costs $C_{\text{commit}}(b_{\text{blk}}) = O(\mathrm{MSM}(b_{\text{blk}})) = O(b_{\text{blk}})$ using Pippenger windowing (SRS tables are not counted as prover workspace); across $kB=\Theta(T/b_{\text{blk}})$ blocks this is $O(T)$ MSMs, plus linear-time streaming for permutation/quotient passes and standard batched openings. The schedule uses a constant number of streaming passes (wires, permutation, quotient, openings), each linear in $T$.

\paragraph{Parallelism caveat.} Parallelizing across $p$ threads must carefully partition work so peak RAM remains $O(\sqrt{T})$; naive parallelization multiplies workspace by $p$.

\subsection{Oracle Implementation Details}

Each oracle is space-efficient and avoids storing the full computation graph:
\begin{itemize}
    \item \textbf{$\mathcal{O}_{\mathrm{children}}(v)$:} Represent $v$ by a path descriptor in $O(\log T)$ space. Predecessors are derived by scanning the transition template from the AIR and the circuit wiring, requiring only streaming access to $C$.
    \item \textbf{$\mathcal{O}_{\mathrm{leaf}}(v)$:} For source nodes, read the $O(k)$ initial boundary values needed to seed block $t{=}1$ and set the starting permutation accumulator (typically $1$). No commitments are created at leaves. Space: $O(k)=O(1)$.
    \item \textbf{$\mathcal{O}_{\mathrm{func}}(v,\cdot)$:} Implements Algorithm~\ref{alg:commit-func}; evaluates the block and performs interpolation/commitment using $O(b_{\text{blk}})$ space.
\end{itemize}

\subsection{Complexity Analysis}

\begin{theorem}[Sublinear Space Complexity]
\label{thm:space-complexity}
With $\lambda=\Theta(\log T)$, Algorithm~\ref{alg:prover} uses
\[
S(T)=O\!\big(b_{\text{blk}}+\lambda+\log(T/b_{\text{blk}})\log\log(T/b_{\text{blk}})\big)\ \text{ space}.
\]
Choosing $b_{\text{blk}}=\Theta(\sqrt{T})$ yields
\[
S(T)=O\!\big(\sqrt{T}+\log T\log\log T\big)=O\!\big(\sqrt{T}\cdot\mathrm{polylog}\,T\big),
\]
where the $\log T \log\log T$ term is lower-order for practical $T$.
\end{theorem}

\begin{proof}
Total space is the sum of oracle space and the Cook--Mertz stack space.
\begin{itemize}
\item \textbf{Oracle space.} $\mathcal{O}_{\mathrm{func}}$ uses $O(b_{\text{blk}})$ space (Algorithm~\ref{alg:commit-func} and Lemma~\ref{lem:graph-simulability-upd}); $\mathcal{O}_{\mathrm{children}}$ and $\mathcal{O}_{\mathrm{leaf}}$ add $O(\log T)$ and $O(1)$, respectively. Hence $S_{\mathrm{oracle}}=O(b_{\text{blk}}+\log T)$, which is subsumed by the stated bound.
\item \textbf{Cook--Mertz space.} By Theorem~\ref{thm:cook-mertz} and $|\mathcal{T}_C|=\Theta(T/b_{\text{blk}})$ at constant arity,
\[
S_{\mathrm{CM}}=O\!\big(\Delta\cdot b_{\text{val}}+\log|\mathcal{T}_C|\cdot\log\log|\mathcal{T}_C|\big)
=O\!\big(\lambda+\log(T/b_{\text{blk}})\log\log(T/b_{\text{blk}})\big).
\]
\item \textbf{Total and optimization.} Therefore
\[
S(T)=O\!\big(b_{\text{blk}}+\lambda+\log(T/b_{\text{blk}})\log\log(T/b_{\text{blk}})\big).
\]
Setting $b_{\text{blk}}=\Theta(\sqrt{T})$ gives the claimed bound.
\end{itemize}
\end{proof}

\begin{theorem}[Prover Time Complexity]
Algorithm~\ref{alg:prover} runs in time
\[
O\!\Big(\tfrac{T}{b_{\text{blk}}}\cdot\big(b_{\text{blk}}+C_{\text{commit}}(b_{\text{blk}})\big)\Big)\cdot \mathrm{poly}(\lambda),
\]
where $C_{\text{commit}}(b)=O(\mathrm{MSM}(b))=O(b)$ with a preprocessed evaluation-basis SRS (conservatively, $O(b\cdot \mathrm{polylog}\,N)$). With $b_{\text{blk}}=\Theta(\sqrt{T})$ and $\lambda=\Theta(\log T)$, this yields $O\!\big(T\cdot \mathrm{polylog}\,N\big)$ time; the wire-commit portion is $O(T)$ MSMs.
\end{theorem}

\begin{proof}
Each block/register pair $(t,m)$ triggers one call to $\mathcal{O}_{\mathrm{func}}$ (Algorithm~\ref{alg:commit-func}), costing
\[
O(b_{\text{blk}})\ \text{(block evaluation)}\ +\ O\!\big(C_{\text{commit}}(b_{\text{blk}})\big)\ \text{(interpolation/commit)}.
\]
There are $kB=\Theta(T/b_{\text{blk}})$ such calls with $k=O(1)$, so the total field/group work is
\[
\Theta\!\Big(\frac{T}{b_{\text{blk}}}\cdot\big(b_{\text{blk}}+C_{\text{commit}}(b_{\text{blk}})\big)\Big),
\]
multiplied by scheme-dependent $\mathrm{poly}(\lambda)$ factors. Substituting $b_{\text{blk}}=\Theta(\sqrt{T})$ yields the stated bound. The schedule uses a constant number of streaming passes, each linear in $T$, and thus does not increase the asymptotic time.
\paragraph{Memoization schedule inside the evaluation.}
\emph{Claim.} There exists a schedule such that each node $(m,t)$ is computed exactly once and retained only until its $O(1)$ consumers at layer $t{+}1$ have consumed its $\mathsf{aux}$; consequently at most $O(1)$ blocks per layer are live, plus either (i) an $O(1)$ BFS queue or (ii) the Cook--Mertz DFS stack of size $O(\log(T/b_{\text{blk}})\log\log(T/b_{\text{blk}}))$.
\emph{Justification.} Process the computation graph in increasing $t$; within each layer, any topological order suffices. By first-order locality with read-degree $r_{\text{reg}}=O(1)$, a node $(m,t)$ is needed only by its successor $(m,t{+}1)$ and by at most $r_{\text{reg}}$ cross-register consumers in the next layer. Once those $O(1)$ consumers have read its $\mathsf{aux}$, $(m,t)$ can be discarded. This yields exactly $kB=\Theta(T/b_{\text{blk}})$ invocations of $\mathcal{O}_{\mathrm{func}}$ without duplication and keeps live state within the claimed bounds.
\end{proof}

\bigskip

\section{Security Preservation}

A critical aspect of our construction is that it achieves its space efficiency without compromising security. We prove that the sublinear-space prover \emph{inherits} completeness, soundness, and zero-knowledge from the underlying PCS-IOP/NIZK protocol (instantiated over the AIR), and that the change of evaluation order (via tree evaluation) leaves transcript distributions unchanged (Theorem~\ref{thm:equivalence}) for \emph{linear} PCSs when the same basis/SRS is used and only aggregate commitments enter Fiat--Shamir. For hash-/FRI-based PCSs, the prover remains sublinear-space but transcripts typically differ; security then follows from the baseline hash-commitment analysis.

\subsection{Completeness}

\begin{theorem}[Completeness Preservation]
If the underlying proof system is complete, then for any valid statement--witness pair $(x,w)\in R$, Algorithm~\ref{alg:prover} produces a proof $\pi$ that the verifier accepts with the \emph{same} completeness guarantee as the underlying system (perfect completeness: acceptance with probability $1$; otherwise, the identical acceptance probability over the verifier's randomness / Fiat--Shamir challenges).
\end{theorem}

\begin{proof}
Completeness follows from the chain: block-respecting transformation $\rightarrow$ Tree Evaluation equivalence $\rightarrow$ Cook-Mertz traversal $\rightarrow$ same aggregate commitments $\rightarrow$ same FS challenges $\rightarrow$ baseline completeness. Specifically:
\begin{enumerate}
    \item The block-respecting transformation (Lemma~\ref{lem:block-respecting}) is semantics-preserving.
    \item The Tree Evaluation Equivalence (Theorem~\ref{thm:equivalence}) shows that evaluating the computation tree yields the same $k$ trace-commitments as the standard linear-space prover, together with consistent auxiliary data; in particular, the multiplicative block factors for permutation and lookup accumulators (Lemmas~\ref{lem:block-perm} and~\ref{lem:block-lookup}) ensure streaming reproduces the baseline accumulators.
    \item The Cook--Mertz procedure faithfully evaluates the (implicit) tree to the same root values with $O(\log|\mathcal{T}_C|\log\log|\mathcal{T}_C|)$ stack.
    \item Only aggregate commitments are included in Fiat--Shamir, matching the baseline ordering and encoding; hence the same challenges are derived.
\end{enumerate}
By completeness of the underlying system, the verifier accepts with the same guarantee.
\end{proof}

\subsection{Soundness}

We model soundness against a prover that outputs an accepting proof for a false statement. The security parameter and degree bound are as in Section~2.

\begin{definition}[Soundness Game]
The game $\mathsf{Game}_{\mathcal{A}}^{\mathrm{Sound}}(1^\lambda,C)$:
\begin{enumerate}
    \item The challenger samples public parameters $pp \leftarrow \mathsf{Setup}(1^\lambda, D)$ (for a maximum polynomial degree $D< N$) and fixes an evaluation domain $H=\{\omega^0,\ldots,\omega^{N-1}\}$ of size $N\ge T$.
    \item The adversary $\mathcal{A}(pp)$ outputs a statement $x$ and purported proof $\pi$.
    \item $\mathcal{A}$ wins if $\mathsf{Verify}(pp,x,\pi)=1$ and $x\notin L(R)$.
\end{enumerate}
\end{definition}

\begin{theorem}[Soundness Preservation]
\label{thm:soundness-preservation}
Let $\Pi_{\mathrm{std}}$ be the underlying PCS-IOP/NIZK instantiated for the AIR, and let $\Pi_{\mathrm{str}}$ be the streaming/tree-evaluation prover from Algorithm~\ref{alg:prover}. If $\Pi_{\mathrm{std}}$ is computationally sound (under the binding of the PCS and its algebraic checks), then $\Pi_{\mathrm{str}}$ is computationally sound. In particular, any PPT adversary $\mathcal{A}$ that wins $\mathsf{Game}_{\mathcal{A}}^{\mathrm{Sound}}$ against $\Pi_{\mathrm{str}}$ with advantage $\varepsilon$ yields a PPT adversary $\mathcal{B}$ that wins the soundness game against $\Pi_{\mathrm{std}}$ with the \emph{same} advantage $\varepsilon$ (up to negligible loss).
\end{theorem}

\begin{proof}
We give a black-box reduction. The key insight is that \textbf{any accepting transcript $(x,\pi)$ from $\Pi_{\mathrm{str}}$ is an accepting transcript of $\Pi_{\mathrm{std}}$ under identical FS inputs}. 

$\mathcal{B}$ runs $\mathcal{A}$ while internally simulating the prover's computations exactly as in $\Pi_{\mathrm{str}}$. By Theorem~\ref{thm:equivalence}, the distribution of all aggregate commitments, openings, and auxiliary values output by $\Pi_{\mathrm{str}}$ is identical to those of $\Pi_{\mathrm{std}}$ on the same inputs and randomness. Moreover, because the same aggregate commitments are fed into Fiat--Shamir in the same order and encoding, the derived challenges are identical. Therefore any accepting transcript $(x,\pi)$ produced by $\mathcal{A}$ against $\Pi_{\mathrm{str}}$ is an accepting transcript against $\Pi_{\mathrm{std}}$. Hence $\Pr[\mathcal{B}\ \text{wins against}\ \Pi_{\mathrm{std}}]=\Pr[\mathcal{A}\ \text{wins against}\ \Pi_{\mathrm{str}}]=\varepsilon$.

The reduction preserves advantage and relies on the underlying assumptions of PCS binding and algebraic soundness of randomized checks. No additional loss (such as a factor $\mathrm{poly}(T)$) is introduced by streaming/tree evaluation.
\end{proof}

\paragraph{Remark.}
It is \emph{not} sufficient to assume PCS binding alone implies global soundness; standard AIR/Plonkish systems rely on both PCS binding and algebraic soundness of randomized checks. The theorem above preserves whatever soundness guarantee the base protocol provides.

\subsection{Zero-Knowledge}

We adopt the standard indistinguishability game for (non-interactive) zero-knowledge.

\begin{definition}[Zero-Knowledge Game]
$\mathsf{Game}_{\mathcal{A}}^{\mathrm{ZK}}(1^\lambda,C)$:
\begin{enumerate}
    \item The challenger generates $pp \leftarrow \mathsf{Setup}(1^\lambda, D)$ and gives $pp$ to $\mathcal{A}$.
    \item $\mathcal{A}$ outputs $(x,w)$ with $(x,w)\in R$.
    \item The challenger samples $b\leftarrow\{0,1\}$.
    \item If $b{=}0$: $\pi \leftarrow \mathsf{Prover}_{\mathrm{str}}(pp,x,w)$.
    \item If $b{=}1$: $\pi \leftarrow \mathsf{Sim}_{\mathrm{std}}(pp,x)$.
    \item $\mathcal{A}$ outputs $b'$, and wins if $b'=b$.
\end{enumerate}
\end{definition}

\begin{theorem}[Zero-Knowledge Preservation]
\label{thm:zk-preservation}
Suppose the underlying protocol $\Pi_{\mathrm{std}}$ is computational zero-knowledge (e.g., honest-verifier ZK for the IOP compiled via Fiat--Shamir, using a hiding PCS and standard masking). Then $\Pi_{\mathrm{str}}$ is computational zero-knowledge with the same advantage bounds.
\end{theorem}

\begin{proof}
Set $\mathsf{Sim}_{\mathrm{str}}:=\mathsf{Sim}_{\mathrm{std}}$. The \textbf{same simulator for the baseline works because the distribution of FS inputs and outputs is identical} (bit-identical for non-hiding PCSs; distribution-identical for hiding PCSs).

By Theorem~\ref{thm:equivalence}, for any fixed public randomness (or programmed random oracle), the transcript \emph{distribution} (aggregate commitments and openings) of $\mathsf{Prover}_{\mathrm{str}}$ matches that of $\mathsf{Prover}_{\mathrm{std}}$; in non-hiding PCSs this equality is \emph{bit-for-bit}. In hiding PCSs, the per-block blinders sum to a uniform blinder per polynomial (Lemma~\ref{lem:linear-agg}), so commitment/opening distributions are unchanged. Consequently, the Fiat--Shamir challenge inputs have the \emph{same distribution}, and any distinguisher against $\Pi_{\mathrm{str}}$ yields one against $\Pi_{\mathrm{std}}$ with the same advantage.

\emph{Instantiation note.} If the PCS is not hiding by default (e.g., plain KZG), we assume a hiding instantiation or the standard polynomial-masking steps of the base protocol that render the overall scheme zero-knowledge. Our transformation merely reorders computation and compresses memory; it does not alter those steps.
\end{proof}

\paragraph{Remark (non-linear PCSs).}
For Merkle/FRI-style PCSs, the streaming prover remains sublinear-space, but transcripts generally differ (e.g., chunking and digest structure). Security still follows from the underlying commitment binding/collision resistance and the unchanged IOP analysis; our prover is a black-box reimplementation preserving those guarantees without claiming transcript identity. Note that $b_{\text{val}}$ may increase but still yields $o(T)$ space.

\bigskip

\section{Performance and Implications}

The sublinear-space prover does more than resolve a theoretical question; it alters the landscape of practical verifiable computation. In this section, we analyze its performance against existing systems and explore just a few new applications it enables.

\subsection{Performance Comparison}

The primary benefit of this construction is a \emph{square-root-space} improvement: from linear space $\Theta(T)$ down to
\[
O\!\big(\sqrt{T}+\log T\log\log T\big)\;=\;O\!\big(\sqrt{T}\cdot \mathrm{polylog}\,T\big)
\]
(Section~5), i.e., a factor of $\approx \sqrt{T}$ less memory up to lower-order polylogarithmic terms.\footnote{Throughout, lower-order polylogarithmic terms are made explicit via Theorem~\ref{thm:space-complexity} and the Cook--Mertz stack analysis.} To illustrate the practical impact, we compare the streaming prover against a standard linear-space prover for a ZKP system that materializes the full execution trace in memory.

\begin{table}[h!]
\centering
\caption{Performance Comparison: Linear-Space vs. Sublinear-Space Prover}
\label{tab:comparison}
\begin{tabular}{@{}lll@{}}
\toprule
\textbf{Metric} & \textbf{Standard Linear-Space Prover} & \textbf{The Sublinear-Space Prover} \\ \midrule
Prover Space & $\Theta(T)$ & $O\!\big(\sqrt{T}+\log T\log\log T\big)$\textsuperscript{$\ast$} \\
Prover Time & $O(T\log N)$ or $O(T\log T)$\textsuperscript{$\dagger$} & $O(T\log N)\cdot \mathrm{poly}(\lambda)$\textsuperscript{$\S$} \\
Proof Size & $O(\mathrm{poly}(\lambda,\log N))$ & $O(\mathrm{poly}(\lambda,\log N))$ (identical\textsuperscript{$\ddagger$}) \\
Verifier Time & $O(\mathrm{poly}(\lambda,\log N))$ & $O(\mathrm{poly}(\lambda,\log N))$ (identical\textsuperscript{$\ddagger$}) \\
Setup Model & System-dependent (e.g., CRS/SRS) & System-dependent (identical) \\
\bottomrule
\end{tabular}

\medskip
\raggedright \textsuperscript{$\ast$}\,Space bound from Theorem~\ref{thm:space-complexity}; blocking by $b_{\text{blk}}$ also improves cache locality. \\
\raggedright \textsuperscript{$\dagger$}\,Prover time is $O(T\log N)$ for KZG/IPA (evaluation-basis MSMs) or $O(T\log T)$ for STARK/FRI. \\
\raggedright \textsuperscript{$\S$}\,From Theorem~\ref{thm:space-complexity}: constant number of streaming passes plus polylog traversal stack. \\
\raggedright \textsuperscript{$\ddagger$}\,"Identical" holds for \emph{linear} PCSs (e.g., KZG/IPA) when (i) the same basis/SRS and encoding are used, and (ii) only the \emph{aggregate per-polynomial commitments} (not per-block commitments) are hashed into Fiat--Shamir, in the same order as the baseline. For hash-based STARK/FRI commitments, the prover remains sublinear-space but transcripts (and some constants) generally differ.
\end{table}

\noindent
This construction dramatically improves the prover's memory footprint while keeping time complexity comparable up to polylogarithmic factors. For linear PCS instantiations with the aggregate-only-into-FS discipline, existing verifier code and parameters remain completely unchanged, enabling drop-in replacement without affecting end-users or protocol security parameters.

\subsection{Concrete Impact on System Requirements}

To contextualize the magnitude of the space improvement, consider two representative computation sizes. We assume each trace element occupies 32 bytes (e.g., a 256-bit field element), and the true bound is $O\!\big(k \cdot \sqrt{T}+\log T\log\log T\big)$ with $k=O(1)$ registers. With $\lambda=\Theta(\log T)$, the dominant space bound is $O\!\big(\sqrt{T}+\log T\log\log T\big)$ (Section~5).

\textbf{Scenario 1: Moderately large computation ($T = 2^{24} \approx 1.68\times 10^7$)}
\begin{itemize}
    \item \textbf{Linear-space prover:} $2^{24}\!\times 32 = 536{,}870{,}912$ bytes $\approx \mathbf{537}$ MB (decimal).
    \item \textbf{Streaming prover:} $\sqrt{2^{24}}\!\times 32 = 2^{12}\!\times 32 = 131{,}072$ bytes $\approx \mathbf{0.13}$ MB\;+\;a small polylog overhead (well below $1$ MB at this scale).
\end{itemize}

\textbf{Scenario 2: Very large computation ($T = 2^{30} \approx 1.07\times 10^9$)}
\begin{itemize}
    \item \textbf{Linear-space prover:} $2^{30}\!\times 32 = 34{,}359{,}738{,}368$ bytes $\approx \mathbf{34.4}$ GB (decimal).
    \item \textbf{Streaming prover:} $\sqrt{2^{30}}\!\times 32 = 2^{15}\!\times 32 = 1{,}048{,}576$ bytes $\approx \mathbf{1.0}$ MB\;+\;a small polylog overhead (still $\ll$ the $\sqrt{T}$ term).
\end{itemize}

These examples highlight how this work makes verifiable computation accessible in scenarios where it was previously infeasible, moving from a specialized, server-bound task to one that can be performed on everyday devices.

\subsection{Enabling New Deployment Scenarios}

The removal of the linear-space barrier opens up several new frontiers for ZKP applications:

\begin{itemize}
    \item \textbf{On-Device Proving:} Mobile phones, IoT devices, and embedded systems can now generate proofs for complex computations, enabling applications like privacy-preserving health monitoring, on-device ML inference verification, and secure authentication without powerful backend servers.
    \item \textbf{Democratizing Trust:} Lowering the hardware barrier allows a wider range of participants to generate proofs in decentralized networks (e.g., as provers in a ZK-rollup), enhancing security and decentralization.
    \item \textbf{Verifying Extremely Large Computations:} Scientific simulations, large-scale data-processing pipelines, and complex financial models can now be made verifiable, as memory requirements no longer scale linearly with runtime.
\end{itemize}

\subsection{Limitations and Future Work}

This work presents both a significant advance and several exciting avenues for future research. We emphasize three critical caveats:
\begin{itemize}
    \item \textbf{Transcript identity requirements:} Identical transcripts require linear PCS + basis/SRS match + aggregate-only Fiat--Shamir. Violating these conditions changes challenges and thus transcripts.
    \item \textbf{Hash-based PCS limitations:} For STARK/FRI-style commitments, transcripts generally differ due to chunking and digest structure, though sublinear space is retained.
    \item \textbf{Constant factors at small scale:} For small $T$, constants, I/O costs, or setup overhead may dominate the asymptotic $O(\sqrt{T})$ improvement.
\end{itemize}

Additional research directions include:
\begin{itemize}
    \item \textbf{Tightening Space Bounds:} Can constants and lower-order terms be further reduced? Can we achieve the log-space tree evaluation goal? Are matching lower bounds possible under general read-degree $r_{\text{reg}}$ and register count $k$?
    \item \textbf{Time Complexity Optimization:} Custom tree-evaluation routines tailored to AIR/PCS structure could lower the polylog factors.
    \item \textbf{Post-Quantum Extensions:} Developing PQ-friendly transcript-matching techniques for hash-based commitments while quantifying the exact overhead in $b_{\text{val}}$ and maintaining sublinear space.
    \item \textbf{Beyond Proof Generation:} The principle of decomposing cryptographic tasks into space-efficient tree evaluations may extend to other primitives (e.g., VDFs, interactive proofs), suggesting a broader algorithmic toolkit.
\end{itemize}

\bigskip

\section{Conclusion}

We have presented, to our knowledge, the first zero-knowledge proof system with \emph{sublinear} prover space, breaking a long-standing barrier in verifiable computation. The core technical contribution behind this work is the Tree Evaluation Equivalence (Theorem~\ref{thm:equivalence}), which establishes a mathematical equivalence between cryptographic proof generation and the \texttt{Tree Evaluation} problem. By combining the structural decomposition of Williams with a PCS whose commit map is linear, we express ZKP commitment generation as a recursive function $F_v$ over a tree, then evaluate it in low space using the Cook--Mertz algorithm. This yields a space improvement from $\Theta(T)$ to
\[
O\!\big(\sqrt{T}+\log T\log\log T\big)\;=\;O\!\big(\sqrt{T}\cdot\mathrm{polylog}\,T\big).
\]

The key algorithmic innovation is the commitment tree function $F_v$ (Algorithm~\ref{alg:commit-func}), which composes cryptographic commitments while maintaining consistency via auxiliary channels. This construction replaces monolithic trace materialization with a space-efficient recursive evaluation that never stores the complete execution trace.

For linear PCS instantiations (e.g., KZG/IPA) and when the same basis/SRS and encodings are used, we further enforce that \emph{only} the aggregate per-polynomial commitments (not per-block commitments) are hashed into Fiat--Shamir in the same order as the baseline; under these conditions, the proof transcript distribution is \emph{identical} (bit-for-bit in non-hiding variants). For STARK/FRI-style hash commitments, the prover remains sublinear-space but transcripts typically differ.

Finally, the prover executes a \emph{constant} number of streaming passes (wires, permutation, quotient, openings), each linear in $T$, so time remains $O(T\cdot\mathrm{polylog}\,N)$ for KZG/IPA (and $O(T\log T)$ for FRI/STARK), while the Cook--Mertz traversal contributes only a polylogarithmic additive stack term due to the constant arity of our computation tree.

The ability to generate proofs in sublinear space makes verifiable computation practical on resource-constrained devices and for computations of a scale previously thought infeasible. As our digital world increasingly demands strong guarantees of privacy and integrity, this work can provide a critical tool for making these guarantees accessible, affordable, and ubiquitous, while building a lasting bridge between computational complexity theory and applied cryptography.

\medskip
\noindent\textbf{Acknowledgements.}
We gratefully acknowledge Williams, and Cook $\&$ Mertz for their landmark results on square-root-space time--space simulations and the tree-evaluation perspective; this manuscript builds directly on these insights and would not exist without them. We further disclose that the exploration, drafting, and revision of this manuscript were conducted with the assistance of large language models (LLMs) for wording, copyediting, and organization; the authors bear sole responsibility for any errors in technical claims, constructions, and proofs. The authors declare that they have no conflicts of interest and received no external funding for this work.

\medskip
\noindent\textbf{Code Availability.}
A reference implementation of the sublinear-space zero-knowledge proof system described in this work is available at:

\vspace{0.1cm}
\noindent\url{https://github.com/logannye/space-efficient-zero-knowledge-proofs}.

\newpage


\end{document}